\documentclass[a4paper, amsfonts, amssymb, amsmath, reprint, showkeys, twoside,superscriptaddress, onecolumn, nofootinbib]{revtex4-2}
\usepackage{amsmath,amsthm}
\usepackage{amssymb}
\usepackage{array}
\usepackage[export]{adjustbox}
\usepackage{rotating}
\usepackage{mathtools}
\usepackage{mathdots}
\usepackage{bm}
\usepackage{dsfont}
\usepackage{xcolor}
\usepackage[caption=false]{subfig}
\usepackage{float}
\makeatletter
\let\newfloat\newfloat@ltx
\makeatother
\usepackage{graphicx}

\usepackage{thmtools}
\usepackage{thm-restate}

\usepackage{mathrsfs}
\usepackage{nccmath}
\usepackage{physics}
\usepackage{comment}
\usepackage[normalem]{ulem}
\usepackage[colorlinks=true, urlcolor=blue, linkcolor=blue, citecolor=blue]{hyperref}
\usepackage[capitalize]{cleveref}
\usepackage[acronym]{glossaries}
\setkeys{glslink}{hyper=false}
\usepackage{makecell}

\usepackage{tikz}
\usetikzlibrary{quantikz2}

\usepackage[compat=0.4]{yquant}
\useyquantlanguage{groups}

\newtheorem{theorem}{Theorem}
\newtheorem{lemma}[theorem]{Lemma}

\usepackage{todonotes}

\newcommand{\bea}{\begin{eqnarray}}
\newcommand{\eea}{\end{eqnarray}}

\newacronym{qft}{QFT}{quantum Fourier transform}
\newacronym{qsp}{QSP}{quantum signal processing}
\newacronym{com}{COM}{centre of mass}
\newacronym[
  plural={DOFs},
  longplural={degrees of freedom}
]{dof}{DOF}{degree of freedom}
\newacronym{lcu}{LCU}{linear combination of unitaries}

\newcommand{\ceil}[1]{\left\lceil #1 \right\rceil}

\newcommand{\arccosh}{\operatorname{arccosh}}
\newcommand{\arcsinh}{\operatorname{arcsinh}}

\begin{document}
\title{Improved constant factors for qubitized Hamiltonian simulation}
\author{Matthew Pocrnic}
\affiliation{Xanadu, Toronto, ON, M5G 2C8, Canada}
\author{Danial Motlagh}
\affiliation{Xanadu, Toronto, ON, M5G 2C8, Canada}
\date{\today}

\begin{abstract}
    Quantum signal processing (QSP) serves as the asymptotically optimal technique for Hamiltonian simulation on a quantum computer. By approximating the time evolution operator via the Jacobi-Anger expansion, the Hamiltonian simulation problem reduces to a problem in polynomial approximation theory: find a sufficient degree-$d$ polynomial series to approximate $e^{-i\tau x}$ on $[-1,1]$ within error $\epsilon$. While $d\in\tilde{\mathcal{O}}(\tau)$ is known to be asymptotically optimal, there exists a gap between state-of-the-art bounds and the optimal constant multiplicative factor, which is approximately equal to 1. Here, we close this gap almost entirely, to the point where possible future improvements will not be of practical significance. Our improvement resides in a careful treatment of the Bessel tail in the Jacobi-Anger series using Kapteyn's and Watson's inequalities, thereby reducing the overhead estimates for all Hamiltonian simulation tasks on quantum computers by a factor of $\approx e/2$.
\end{abstract}
\maketitle

\section{Introduction}\label{sec:intro}
Hamiltonian simulation stands as one of the most promising applications of quantum computers, with quantum signal processing (QSP) serving as the asymptotically optimal algorithm, achieving a linear scaling in evolution time $t$. Not only does it enable the study of quantum dynamics \cite{low2017optimal, mukhopadhyay2024quantum, lang2026quantum, motlagh2025quantum, stetina2025first}, it serves as a subroutine for a host of algorithmic tasks such as matrix inversion \cite{morales2026quantum, jennings2023efficient}, phase estimation \cite{loaiza2025nonlinear}, solving differential equations \cite{low2025optimal, li2026quantum, an2023linear, penuel2024detailed}, and time evolving under Lindbladians \cite{pocrnic2025quantum, garg2025simulating, ding2024simulating}.
Hamiltonian simulation via QSP involves building an $\epsilon$-accurate polynomial approximation to the time-evolution operator, $e^{iHt}$, via the Jacobi-Anger expansion. This is achieved by repeated calls to the qubitized walk operator $W = e^{\pm i \arccos(H)}$ constructed from a block-encoding of the Hamiltonian
\begin{equation}
    \|e^{iHt} - \sum_{n=-d}^{d} i^n J_n(t) \,W^n \| \leq \epsilon,
\end{equation}
where $J_n(t)$ is the $n^{th}$ Bessel function of the first kind and $d\in\mathcal{O}(t + \log(1/\epsilon))$. The original construction \cite{low2017optimal} achieves this using $2d$ calls to $W$. This was further improved to $d+2$ with the advent of generalized quantum signal processing (GQSP) \cite{motlagh2024generalized, berry2024doubling}, thereby halving the number of calls to $W$. This reduced the multiplicative constant factor in the number of calls to $W$ as a function of $d$ to $1$. However, the multiplicative constant factor in $d$ as a function of $t$ has remained suboptimal. The first constant factor bound we are aware of comes from Ref. \cite{gilyen2019quantum}, which introduced the Quantum Singular Value Transform. This work bounded the truncation degree by $d = \ceil{2\alpha t + 3\ln(12/\epsilon)}$ (Corollary 62), where $\alpha\geq \|H\|$ is the 1-norm of the Hamiltonian induced by the block encoding. This was later improved in Ref. \cite{jennings2023efficient} (version 1, Lemma 5) to $d=\ceil{\frac{e}{2}\alpha t +  \ln \left (\frac{2c}{\epsilon}\right ) }$, with $c=4(\sqrt{2\pi}e^{1/13})^{-1} \approx 1.47762$, via a more careful treatment involving the Lambert-$W$ function. In practical problems $\alpha t \gg \ln(1/\epsilon)$, therefore finding better constant factors multiplying this term results in immediate algorithmic improvements. Numerical computation of the error contributions from Bessel tail bound suggest a leading factor of approximately $1$ should be possible, but has remained elusive thus far. \\

 In this work, we close this gap by proving constant factor bounds that approach a leading constant of $1$ in practical parameter regimes, thereby reducing the overhead estimates for all Hamiltonian simulation tasks on quantum computers by a factor of $\approx e/2$. We achieve this improvement by deriving two bounds using each of Kapteyn's inequality \cite{reif2024lower} and Watson's inequality \cite{watson1922treatise} to bound the Bessel tail in the Jacobi-Anger series by a majorant power series. We further demonstrate the tightness of our bounds in \cref{fig:degree}, by showing that our analytical bound rapidly converges to the numerically computed magnitude of the true Bessel tail. When combined with the results of GQSP \cite{motlagh2024generalized}, this implies that the multiplicative constant factor in the number of calls to $W$ as a function of $t$ for implementing the time-evolution operator, $e^{iHt}$, is effectively equal to $1$.

\section{Main Result}
The main result of this work is a rigorous constant factor bound on the number of queries to the walk operator $W = e^{\pm i\, \arccos(H/\alpha)}$ required to approximate the time evolution operator $e^{-iHt}$ within error $\epsilon$. The standard approach in the literature is to use the Jacobi-Anger series
\begin{equation}
    e^{ix\cos(\theta)} = \sum_{k=-\infty}^\infty i^k J_k(x)e^{ik\theta}, 
\end{equation}
where $e^{ik\theta}$ is the scalar analog of the walk operator, and $J_k(x)$ are Bessel functions of the first kind. The Bessel property $J_{-k}(x) = (-1)^kJ_k(x)$ further implies the expansion 
\begin{equation}
    e^{ix\cos(\theta)} = J_0(x)+2\sum_{k=1}^\infty i^k J_k(x)\cos(k\theta).
\end{equation}
Written in the walk operator language, $e^{i\theta}\to e^{i\arccos(H/\alpha)}$, we obtain 
\begin{equation}
    e^{itH} = J_0(\alpha t)+2\sum_{k=1}^\infty i^k J_k(\alpha t)T_k(H/\alpha),
\end{equation}
where $\alpha \geq \|H\|$ is the block encoding 1-norm. To implement an $\epsilon$-precise finite polynomial series approximation, we need to truncate the result on the right hand side at some target degree $d$, which will then dictate the number of queries to the walk operator $W$ and by extension, to the block encoding $U_H$ \cite{berry2024doubling}.
Our resulting bounds on block encoding queries are therefore corollaries to Lemmas \ref{lem:degree} and \ref{lem:watson_degree}, which bound the sufficient choice of truncation degree in the Jacobi-Anger series. 

\begin{figure}[!t]
    \centering
    \includegraphics[width=0.85\textwidth]{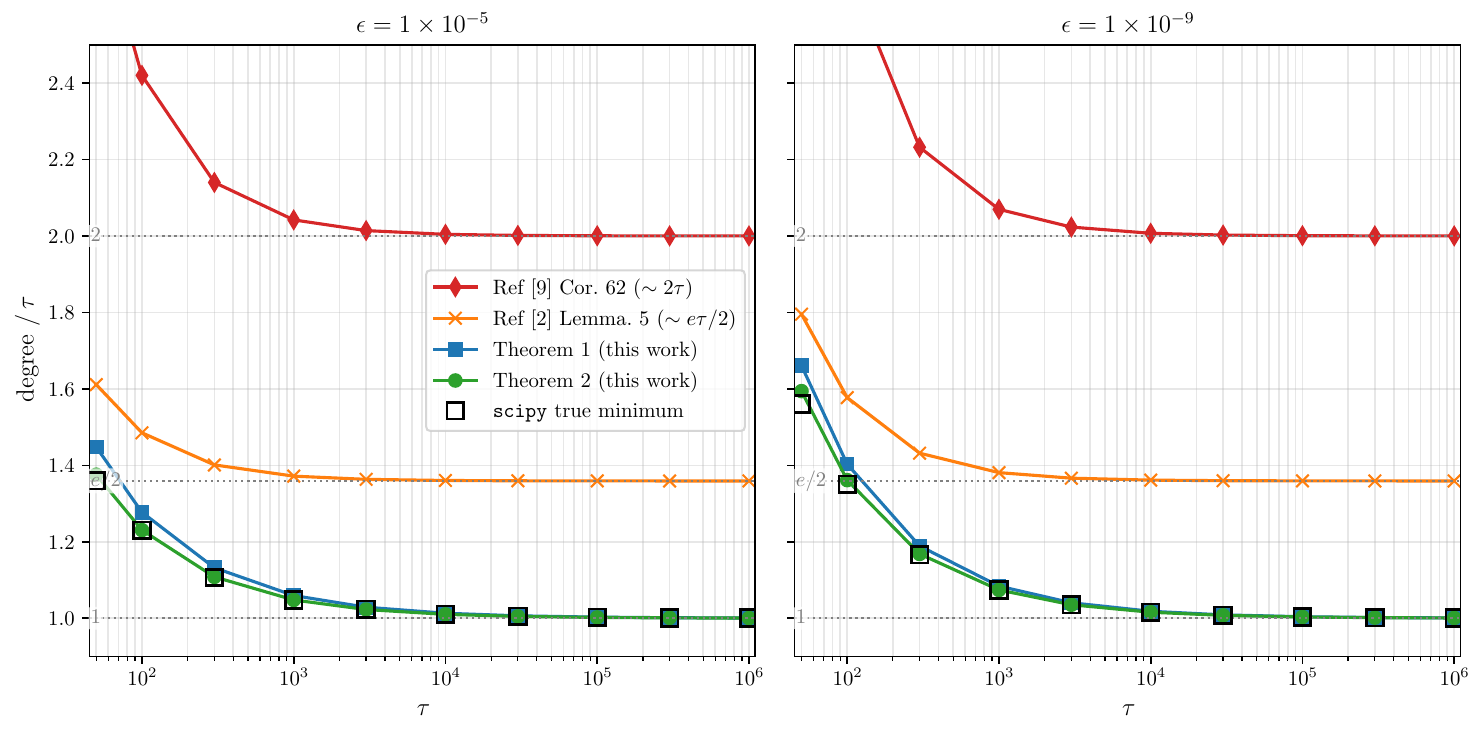}
    \caption{Comparison of degree bounds with $\tau:=\alpha t$ from Theorems \ref{thm:hamsim} and \ref{thm:hamsim2} to the exact Bessel tail, Corollary 62 from Ref. \cite{gilyen2019quantum}, and Lemma 5 from \cite{jennings2023efficient} version 1. We plot the exact Bessel tail to machine precision via the clear boxes using SciPy \cite{virtanen2020scipy} to show that the bound is tight and to further illustrate that the gap to the optimal result becomes negligible with increasing $\tau$. To contextualize the results, for the simulation of modestly sized organic chemical reactions in first quantization, the product $\alpha t$ approximately satisfies $\alpha t \approx 10^4 \text{-}10^6$ \cite{pocrnic2026efficient}.}
    \label{fig:degree}
\end{figure}

\begin{theorem}[Walk operator queries] \label{thm:hamsim} Given access to $U_H$, a $(\alpha, m, 0)$ block encoding of a Hamiltonian $H\in \mathbb{C}^{2^n\times 2^n}$, it is possible to construct a $(1, m+2, \epsilon)$ block encoding of $e^{-iHt}$ using $d+2$ queries to $U_H$ or $U_H^\dagger$ where 
\begin{align}
    d &= \left\lceil \alpha t \cosh \left ( \left [\frac{3}{\alpha t}\ln\left (\frac{2}{\epsilon (1-e^{-\theta_{0}(\alpha t, \epsilon)})}\right )\right ]^{1/3} \right )\right \rceil \\ &= \alpha t +\frac{(\alpha t)^{1/3}}{2}\left(3\ln \left (\frac{2}{\epsilon(1-e^{-\theta_0})}\right) \right )^{2/3} + \tilde{O}\left (\frac{1}{(\alpha t)^{1/3}} \right),
\end{align}
 and $\theta_{0}(\alpha t, \epsilon) := \left [\frac{3}{\alpha t}\ln\left (\frac{2}{\epsilon }\right )\right ]^{1/3}$. This holds $\forall \, d >\alpha t$, $t \in \mathbb{R_+}$.
\end{theorem}
\noindent
Therefore, in the large $\alpha t$ regime, the leading constant asymptotically converges to 1. This is numerically verified in Figure \ref{fig:degree}. Specifically we plot $d(\tau, \epsilon)/\tau$, and show that for our bound this approaches 1 in the practical $\tau$ regime. Our proof of the main result relies on Lemma \ref{lem:error} and Lemma \ref{lem:degree}, which we prove in Section \ref{sec:proofs}. This result is effectively optimal and the practical bound for most applications, speeding up a variety of works based on qubitized dynamics \cite{pocrnic2026efficient, eklund2026end, da2025comprehensive, pocrnic2025constant, rubin2024quantum, kharazi2026quantum}. \\

We note that the leading order 1 is also obtained via an asymptotic formula for the Bessel tail based on the Airy function in Ref. \cite{babbush2019quantum}, however, while it is predictive for large $\alpha t$ (matching numerical behaviour in this regime), it is not a rigorous bound on the polynomial degree, and fails to accurately describe the behaviour for smaller values of $\alpha t$. On this note, the result claimed in Theorem \ref{thm:hamsim} is loose by a logarithmic factor which can become notable in this regime. \\

Given the large effort in the field to reduce $\alpha$ over the years \cite{dutkiewicz2026spectral, loaiza2023block, king2026quantum, lee2021even}, in addition to applications that simulate the qubitized dynamics for a short time-step $\Delta t$ as in Ref. \cite{haah2021quantum}, which uses qubitization to build the exponentials in subsequently constructed product formulae, it is of interest to tightly bound the entire range. By employing Watson's inequality \cite{watson1922treatise}, the bound based on Kapetyn is further tightened in the Theorem \ref{thm:hamsim2}, which follows from Lemma \ref{lem:watson_degree} proven in Section \ref{sec:proofs}. As an added bonus, this rigorous bound is tighter than the asymptotic expression based on the Airy function derived in Ref. \cite{babbush2019quantum}.

\begin{theorem}[Tightened walk operator queries] \label{thm:hamsim2} Given access to $U_H$, a $(\alpha, m, 0)$ block encoding of a Hamiltonian $H\in \mathbb{C}^{2^n\times 2^n}$, it is possible to construct a $(1, m+2, \epsilon)$ block encoding of $e^{-iHt}$ using $d+2$ queries to $U_H$ or $U_H^\dagger$ where 
\begin{equation}
    d(\tau, \epsilon)= \left\lceil \tau \cosh \left (\left [\frac{3}{\tau}\ln\left (\frac{2}{\epsilon (1-e^{-\theta_{\rm lo}}) \sqrt{2\pi \tau \sinh(\theta_{\rm lo})}}\right )\right ]^{1/3} \right )\right \rceil,
\end{equation}
 given we define $\theta_{0} := \left [\frac{3}{\alpha t}\ln\left (\frac{2}{\epsilon }\right )\right ]^{1/3}$, $\theta_{\rm up} := \left[\frac{3}{\tau}\ln\left(\frac{2}{\epsilon(1-e^{-\theta_0})}\right)\right]^{1/3}$, and $\theta_{\rm lo} := \left [\frac{3}{\tau}\ln\left (\frac{2}{\epsilon (1-e^{-\theta_{\rm up}}) \sqrt{2\pi \tau \sinh(\theta_{\rm up})}}\right )\right ]^{1/3}$, and satisfy $d\geq \sqrt{(\alpha t)^2+1/4\pi^2}$ with $t\in \mathbb{R}_+$.
\end{theorem}

\noindent This bound more closely tracts the error over the entire range as show in Figure \ref{fig:degree}. 

\section{Proofs} \label{sec:proofs}
In this section we will use $\tau:=\alpha t$ for compactness. In order to prove the main result, we rely on a tail bound on an infinite series of Bessel functions. To bound a single Bessel function we make use of Kapteyn's inequality \cite{reif2024lower}:
\begin{equation}
|J_k(\tau)|\ \le\
\left(\frac{(\tau/k)\,e^{\sqrt{1-(\tau/k)^2}}}
{1+\sqrt{1-(\tau/k)^2}}\right)^{k}, \: \: \: (0\leq \tau \leq k).
\label{eq:kapteyn}
\end{equation}
This inequality specifically bounds the Bessel tail in the regime of our Jacobi-Anger tail bound, in other words, where the degree $k\geq \tau$. We bound said tail sum of Bessel functions in the following Lemma.

\begin{lemma}[Jacobi-Anger Truncation Error] \label{lem:error}
    Given the following definition of the truncation error $\mathcal T_\epsilon(\tau, d)$
    \begin{equation} \label{eq:max_dev}
        \mathcal T_\epsilon(\tau, d) := \max_{x\in [-1,1]}\left |e^{i\tau x} - J_0(\tau) - 2\sum_{k=1}^{d-1}i^k J_k(\tau) T_k(x) \right |,
    \end{equation}
    with $d \in \mathbb{Z}_+$ and $\tau \in \mathbb{R}_+$, then the following bound holds $\forall \,\tau < d$
    \begin{equation} \label{eq:transcendental}
        \mathcal T_\epsilon(\tau, d) \leq 2\frac{e^{-d \, \arccosh(d/\tau)+\sqrt{d^2 - \tau^2}}}{1-e^{- \arccosh \left (\frac{d}{\tau}\right)}}.
    \end{equation}
\end{lemma}
\begin{proof}
Our starting point for analyzing the necessary polynomial degree is the following:
\begin{align}
    \max_{x\in [-1,1]}\left |e^{i\tau x} - J_0(\tau) - 2\sum_{k=1}^{d-1} i^k J_k(\tau) T_k(x) \right | &\leq 2\max_{x\in [-1,1]} \sum_{k=d}^\infty |J_k(\tau)T_k(x)| \\
    &\leq 2\sum_{k=d}^\infty |J_k(\tau)|,
\end{align}
obtained using the fact that the Chebyshev polynomials are bounded between $[-1,1]$ on the interval. Now since we know the polynomial degree $d$ is necessarily greater than $\tau$, we apply Kapteyn's inequality from Equation \eqref{eq:kapteyn}
\begin{align}
    2\sum_{k=d}^\infty|J_k(\tau)| \leq 2\sum_{k=d}^\infty \left(\frac{(\tau/k)\,e^{\sqrt{1-(\tau/k)^2}}}
{1+\sqrt{1-(\tau/k)^2}}\right)^{k},
\end{align}
and we are tasked with bounding the tail of the infinite series. This task becomes more amenable by transforming to a hyperbolic form:
\begin{align}
\left(\frac{e^{\sqrt{1-(\tau/k)^2}}}
{(k/\tau)+\sqrt{(k/\tau)^2-1}\,}\right)^{k} 
&= \frac{e^{\sqrt{k^2-\tau^2}}}{e^{k \arccosh(k/\tau)}}\\
&= \exp \left (-k \arccosh(k/\tau) + \sqrt{k^2-\tau^2}\right )\\
&= \exp\left (- \int_\tau^k  \arccosh(s/\tau)\, \mathrm{d}s \right ) ,
\end{align}
where the final line comes from the definition of the hyperbolic cosine integral and can be verified by direct computation. Substituting this into our tail bound leaves us with 

\begin{align} \label{eq:kapetyn_integral}
    2\sum_{k=d}^\infty |J_k(\tau)| \leq 2\sum_{k=d}^\infty \exp\left (- \int_\tau^k  \arccosh(s/\tau)\, \mathrm{d}s \right ) := 2\sum_{k=d}^\infty K_\tau (k). 
\end{align}
We proceed by using the standard technique of bounding by a majorant power series.
\begin{align}
    \frac{K_\tau(k+1)}{K_\tau(k)} &= \exp \left (- \int_\tau^{k+1}  \arccosh(s/\tau)\, \mathrm{d}s + \int_\tau^k  \arccosh(s'/\tau)\, \mathrm{d}s' \right ) \\
    & = \exp \left (- \int_k^{k+1}  \arccosh(s/\tau)\, \mathrm{d}s \right ) \\
    & \leq \exp \left (- \arccosh \left (\frac{k}{\tau}\right) \right ) \\
    & \leq \exp \left (- \arccosh \left (\frac{d}{\tau}\right) \right ).
\end{align}
In the second to last line, we take advantage of the fact that $\arccosh(x)$ is monotonically increasing, and lower-bound the integral using its minimum value over the interval. This lower bound paired with that fact that the function is an argument of the negative exponential yields an upper-bound. Using this maximum decay ratio of the terms, we can bound the tail of the Bessel sum with a majorant power series 
\begin{align}
    2\sum_{k=d}^\infty |J_k(\tau)| &\leq 2 \sum_{k=d}^\infty \left(\frac{(\tau/k)\,e^{\sqrt{1-(\tau/k)^2}}}
{1+\sqrt{1-(\tau/k)^2}}\right)^{k} \\
& \leq 2 \,K_\tau(d) \sum_{k=d}^\infty\exp \left (- \arccosh \left (\frac{d}{\tau}\right) \right )^{k-d} \\
&= 2\frac{e^{-\int_\tau^d \arccosh(s/\tau)\mathrm{d}s}}{1-e^{{- \arccosh \left (\frac{d}{\tau}\right)}}} \\
&= 2\frac{e^{-d \, \arccosh(d/\tau)+\sqrt{d^2 - \tau^2}}}{1-e^{- \arccosh \left (\frac{d}{\tau}\right)}},
\end{align}
where in the last two lines we sum the geometric series and evaluate the integral. 
\end{proof}

The result of Lemma \ref{lem:error} is a transcendental equation that has no known solutions. To find the polynomial degree we can either solve this equation numerically, or can further linearize the bound into a soluble expression. Since we are interested in a rigorous bound on the degree, we perform the latter. By linearizing the bound, we will obtain an analytic expression for $d(\tau, \epsilon)$ which is necessary to prove the main result. We provide said analytic bound on $d(\tau, \epsilon)$ in the following Theorem.

\begin{lemma}[Polynomial degree bound via Kapteyn] \label{lem:degree}
Let $\tau \in \mathbb{R}_+$, $\epsilon \in (0,1/2)$, and define 
\begin{equation}
    \theta_{0}(\tau, \epsilon) := \left [\frac{3}{\tau}\ln\left (\frac{2}{\epsilon }\right )\right ]^{1/3}.
\end{equation}
Then choosing the Jacobi-Anger truncation degree
    \begin{equation}
        d(\tau, \epsilon)= \left\lceil \tau \cosh \left ( \left [\frac{3}{\tau}\ln\left (\frac{2}{\epsilon (1-e^{-\theta_{0}(\tau, \epsilon)})}\right )\right ]^{1/3} \right )\right \rceil
    \end{equation}
    suffices to approximate the function $e^{i\tau x}$ to error $\epsilon$, where $\epsilon$ is an upper bound on the maximum deviation $\mathcal T_\epsilon(\tau,d)$ defined in Equation \eqref{eq:max_dev}. This choice of $d(\tau, \epsilon)$ holds $\forall \, d > \tau$.
\end{lemma}
\begin{proof}
To linearize this expression, let 
\begin{equation}
    \theta := \arccosh(d/\tau), \quad \theta>0
\end{equation}
and set $d = \tau \cosh(\theta)$. Subbing into Equation \eqref{eq:transcendental}
\begin{align}
    \mathcal{T}_\epsilon(\tau, d) \leq 2\frac{e^{-d \, \arccosh(d/\tau)+\sqrt{d^2 - \tau^2}}}{1-e^{- \arccosh \left (\frac{d}{\tau}\right)}} = 2\frac{e^{-\tau[\theta \cosh(\theta) - \sinh(\theta)]}}{1-e^{-\theta}}.
\end{align}
Now Taylor expanding the argument of the exponential in the numerator we have $\theta \cosh(\theta) - \sinh(\theta) = \sum_{j=1}^\infty \frac{2j \, \theta^{2j+1}}{(2j+1)!} \geq \frac{\theta^3}{3}$, since the series is strictly positive. This in turn upper bounds the exponential and yields 

\begin{align}
    2\frac{e^{-\tau[\theta \cosh(\theta) - \sinh(\theta)]}}{1-e^{-\theta}} \leq 2\frac{e^{-\tau \theta^3/3}}{1-e^{-\theta}}.
\end{align}
Now enforce that this is bounded by $\epsilon$:

\begin{align}
     &2 \,\frac{e^{-\tau \theta^3/3}}{1-e^{-\theta}} \leq \epsilon \\
     & \theta \geq \left [\frac{3}{\tau}\ln\left (\frac{2}{\epsilon (1-e^{-\theta})}\right )\right ]^{1/3}. 
\end{align}
This is still not solvable for $\theta$, however, we can proceed by showing that there exists a unique fixed point that saturates this inequality, and find an upper bound on this fixed point that still satisfies the inequality. First define 
\begin{equation}
    f(\theta):= \left [\frac{3}{\tau}\ln\left (\frac{2}{\epsilon (1-e^{-\theta})}\right )\right ]^{1/3},
\end{equation}
and note that $f(\theta)$ is strictly decreasing and strictly positive on the domain $\theta\in (0, \infty)$. On the other hand, $\theta$ is strictly increasing and positive on this domain. In the limit $\theta \to 0^+$ we have $\theta <f(\theta)$ and in the limit $\theta \to \infty$ we have $\theta > f(\theta)$ which holds $\forall \, \tau, \epsilon$ as defined in the Theorem statement. These conditions imply a unique fixed point $\theta^\star$ such that $f(\theta^\star) = \theta^\star$. Now if we find a lower bound $\theta_0$ such that $ \theta \geq \theta^\star \geq \theta_0$ (since the fixed point is the smallest such $\theta$ that satisfies the inequality by the monotonicity and positivity of the functions) to upper-bound the fixed point, we have that $\theta \geq f(\theta_0)$, thus satisfying the inequality. This is true because $f(\theta)$ is a strictly decreasing function, and applying it reverses the inequality like so: $\theta\geq \theta^\star \implies f(\theta^\star)\leq f(\theta_0)$. Therefore, let 
\begin{equation}
    \theta_0 := \lim_{\theta \to \infty} f(\theta) = \left [\frac{3}{\tau}\ln\left (\frac{2}{\epsilon }\right )\right ]^{1/3},
\end{equation}
and using the property $\theta \geq f(\theta_0)$ gives 
\begin{equation}
    \theta \geq  \left [\frac{3}{\tau}\ln\left (\frac{2}{\epsilon (1-e^{-\theta_0})}\right )\right ]^{1/3},
\end{equation}
$\forall \, \theta \geq \theta^\star$, where $\theta^\star$ is the fixed point and the minimum $\theta$ that satisfies the inequality. Now substituting $\theta = \arccosh(d/\tau)$ and solving for $d$ gives 
\begin{equation}
    d \geq \tau \cosh \left ( \left [\frac{3}{\tau}\ln\left (\frac{2}{\epsilon (1-e^{-\theta_{0}})}\right )\right ]^{1/3} \right ).
\end{equation}
Enforcing $d$ to be an integer and choosing the $d$ that saturates the inequality completes the proof. 
\end{proof}
\noindent
We remark that this strategy can be recursed to achieve increasingly accurate approximations of the fixed point by iteratively substituting $\theta_1, \theta_2...$; however, we find that numerically this makes no practical difference in the standard practical regime where $\tau \gg \log(1/\epsilon)$. Theorem \ref{thm:hamsim} then immediately follows from the Lemma by applying the art of \cite{berry2024doubling}.\\

The results so far are loose by a logarithmic factor and can be further tightened via replacing Kapteyn's inequality with Watson's inequality [Eq. 9, Ref. \cite{watson1922treatise} pg. 255]: 
\begin{equation} \label{eq:watson}
    J_\nu(\nu z) \leq \frac{e^{-\nu F(z)}}{(1-z^2)^{1/4}\,(2\pi\nu)^{1/2}}, \qquad F(z) = \ln\left(\frac{1+\sqrt{1-z^2}}{z}\right) - \sqrt{1-z^2},
\end{equation}
which holds for all $\nu > 0 $ and $0<z\leq 1$. For the case where the Bessel degree is substantially larger than its argument (which is the case of the terms in the tailbound since $k\to \infty$), Watson's inequality is significantly tighter than Kapteyn's due to the additional terms in the denominator. This bound has been previously utilized to derive an asymptotic form of the Bessel tail error in 3-dimensional scattering theory \cite{carayol2004error}. Now, setting $\nu = k$ and $z=\tau/k$ (such that $k\geq \tau$), we have $k F(\tau/k) = k\arccosh(k/\tau) - \sqrt{k^2-\tau^2}$ and $(1-z^2)^{1/4}(2\pi k)^{1/2} = (2\pi \sqrt{k^2-\tau^2})^{1/2}$, so that Equation \eqref{eq:watson} bounds the tail term-wise
\begin{equation} \label{eq:watson_term}
    2\sum_{k=d}^\infty|J_k(\tau)| \leq \sum_{k=d}^\infty\frac{2}{(2\pi \sqrt{k^2-\tau^2})^{1/2}}\, e^{-k\,\arccosh(k/\tau)+\sqrt{k^2-\tau^2}}.
\end{equation}
Since $k\geq \tau$ for all indices in the sum, the leading term is strictly decreasing and the sum can be trivially upper-bounded by taking its value at $k=d$ giving 
\begin{equation} \label{eq:improved_tail}
    2\sum_{k=d}^\infty|J_k(\tau)| \leq \frac{2}{(2\pi \sqrt{d^2-\tau^2})^{1/2}}\,\sum_{k=d}^\infty  e^{-k\,\arccosh(k/\tau)+\sqrt{k^2-\tau^2}} = \frac{2}{(2\pi \sqrt{d^2-\tau^2})^{1/2}}\, \sum_{k=d}^\infty K_\tau(k),
\end{equation}
where $K_\tau(k)$ is the Kapteyn factor defined in Equation \eqref{eq:kapetyn_integral}. We can now use the results of the Lemma \ref{lem:error} and Lemma \ref{lem:degree} to prove a new tighter bound in Lemma \ref{lem:watson_degree}. 
\begin{lemma}[Polynomial degree bound via Watson] \label{lem:watson_degree}
Let $\tau \in \mathbb{R}_+$, $\epsilon \in (0,1/2)$, and define 
\begin{equation}
    \theta_{0} := \left[\frac{3}{\tau}\ln\left(\frac{2}{\epsilon}\right)\right]^{1/3}, \qquad
    \theta_{\rm up} := \left[\frac{3}{\tau}\ln\left(\frac{2}{\epsilon(1-e^{-\theta_0})}\right)\right]^{1/3}, \qquad
    \theta_{\rm lo} := \left [\frac{3}{\tau}\ln\left (\frac{2}{\epsilon (1-e^{-\theta_{\rm up}}) \sqrt{2\pi \tau \sinh(\theta_{\rm up})}}\right )\right ]^{1/3}.
\end{equation}
Then choosing the Jacobi-Anger truncation degree
    \begin{equation}
        d(\tau, \epsilon)= \left\lceil \tau \cosh \left (\left [\frac{3}{\tau}\ln\left (\frac{2}{\epsilon (1-e^{-\theta_{\rm lo}}) \sqrt{2\pi \tau \sinh(\theta_{\rm lo})}}\right )\right ]^{1/3} \right )\right \rceil
    \end{equation}
    suffices to approximate the function $e^{i\tau x}$ to error $\epsilon$, where $\epsilon$ is an upper bound on the maximum deviation $\mathcal T_\epsilon(\tau,d)$ defined in Equation \eqref{eq:max_dev}. This choice of $d(\tau, \epsilon)$ holds $\forall \, d\geq \sqrt{\tau^2+1/4\pi^2}$.
\end{lemma}
\begin{proof}
    Our starting point is Equation \eqref{eq:improved_tail}. We know from Lemma \ref{lem:error} that
    $\sum_{k=d}^\infty K_\tau(k) \leq \frac{\exp\left ({-d \, \arccosh(d/\tau)+\sqrt{d^2 - \tau^2}}\right)}{1-\exp\left({- \arccosh \left (\frac{d}{\tau}\right)}\right)}$, thus substituting gives 
    \begin{equation}
        2\sum_{k=d}^\infty|J_k(\tau)|\leq \frac{2}{(2\pi \sqrt{d^2-\tau^2})^{1/2}}\,\frac{e^{-d \, \arccosh(d/\tau)+\sqrt{d^2 - \tau^2}}}{1-e^{- \arccosh \left (\frac{d}{\tau}\right)}}.
    \end{equation} 
    Once again making the substitution $\theta = \arccosh(d/\tau)$ this becomes 
    \begin{equation}
        2\sum_{k=d}^\infty|J_k(\tau)|\leq \frac{2}{\sqrt{(2\pi \tau \sinh(\theta))}}\,\frac{e^{-\tau[\theta \cosh(\theta) - \sinh(\theta)]}}{1-e^{-\theta}} \leq \frac{2}{\sqrt{(2\pi \tau \sinh(\theta))}}\,\frac{e^{-\tau \theta^3/3}}{1-e^{-\theta}} \leq \epsilon,
    \end{equation}
    where we have once again used that $\theta \cosh(\theta) - \sinh(\theta) = \sum_{j=1}^\infty \frac{2j \, \theta^{2j+1}}{(2j+1)!} \geq \frac{\theta^3}{3}$. Rearranging for $\theta$ we obtain a similar fixed point condition to that in the proof of Lemma \ref{lem:degree}:
    \begin{equation} \label{eq:fixed_point}
        \frac{2}{\sqrt{(2\pi \tau \sinh(\theta))}}\,\frac{e^{-\tau \theta^3/3}}{1-e^{-\theta}} \leq \epsilon \iff \theta \geq  \left [\frac{3}{\tau}\ln\left (\frac{2}{\epsilon (1-e^{-\theta}) \sqrt{2\pi \tau \sinh(\theta)}}\right )\right ]^{1/3}  := g(\theta). 
    \end{equation}
    Since $\theta\in (0, \infty)$ we have that both sides are strictly positive and $g(\theta)$ is strictly decreasing in $\theta$, further implying the existence of a unique fixed point $\phi^\star$ (different than the fixed point $\theta^\star$ of Lemma \ref{lem:degree}) where the fixed point is the smallest possible point that satisfies the inequality. To understand the relationship between these fixed points observe that $g(\theta) = f(\theta)$ when ${\sqrt{2\pi \tau \sinh(\theta)}} = 1$ (as they only differ by this factor in the logarithm). Rearranging this expression and substituting $\theta = \arccosh(d/\tau)$ we obtain $\cosh(\arcsinh(1/2\pi\tau)) = d/\tau$. Using the hyperbolic identity $\cosh(\arcsinh(x)) = \sqrt{1+x^2}$, we simplify the expression to obtain that $d=\sqrt{\tau^2+1/4\pi^2}$ at the point of intersection $\theta_{\rm poi}$ of these functions. Therefore, $g(\theta) \leq f(\theta) \, \forall \, d\geq \sqrt{\tau^2+1/4\pi^2}$. Requiring $d \in \mathbb{N}$ and $d>\tau$ will satisfy this requirement in almost all practical cases, however, we demand $d\geq \sqrt{\tau^2+1/4\pi^2}$ hold in the theorem statement to be rigorous. Enforcing this inequality and noting that $g(\theta)\leq f(\theta)\, \forall \, \theta \geq \theta_{\rm poi}$, we necessarily have $\phi^\star \leq \theta^\star$ as expected from the tighter nature of Watson's bound.
    
    From Lemma \ref{lem:degree}, if we define the upper bound $\theta_{\rm up}: = \left [\frac{3}{\tau}\ln\left (\frac{2}{\epsilon (1-e^{-\theta_0})}\right )\right ]^{1/3} \geq \theta^\star \geq \phi^\star$  we proved that 
    \begin{equation}
        \frac{2}{\sqrt{(2\pi \tau \sinh(\theta_{\rm up}))}}\,\frac{e^{-\tau \theta_{\rm up}^3/3}}{1-e^{-\theta_{\rm up}}} \leq 2\,\frac{e^{-\tau \theta_{\rm up}^3/3}}{1-e^{-\theta_{\rm up}}} \leq  \epsilon,
    \end{equation}
    where the LHS is trivially true. By the fixed point condition, this implies that $\theta_{\rm up} \geq \phi^\star$ since $\phi^\star$ is the smallest value that satisfies the inequality. By definition of the fixed point we have 
    \begin{equation}
        \phi^\star =g(\phi^\star) = \left [\frac{3}{\tau}\ln\left (\frac{2}{\epsilon (1-e^{-\phi^\star}) \sqrt{2\pi \tau \sinh(\phi^\star)}}\right )\right ]^{1/3} \geq \left [\frac{3}{\tau}\ln\left (\frac{2}{\epsilon (1-e^{-\theta_{\rm up}}) \sqrt{2\pi \tau \sinh(\theta_{\rm up})}}\right )\right ]^{1/3} := \theta_{\rm lo},
    \end{equation}
    where a lower bound is obtained by evaluating the decreasing function $g(\theta_{\rm up})$ at the larger argument ($\theta_{\rm up} \geq \phi^\star$). Since $\theta_{\rm lo} \leq \phi^\star \implies \theta \geq g(\theta_{\rm lo}) \geq \phi^\star$, we can solve the inequality $\theta \geq g(\theta_{\rm lo})$ since $\theta_{\rm lo}$ is independent of $\theta$ itself. Using Equation \eqref{eq:fixed_point} we now solve for $d$ and obtain:

    \begin{equation}
        d \geq  \tau \cosh \left (\left [\frac{3}{\tau}\ln\left (\frac{2}{\epsilon (1-e^{-\theta_{\rm lo}}) \sqrt{2\pi \tau \sinh(\theta_{\rm lo})}}\right )\right ]^{1/3} \right )
    \end{equation}
    Enforcing $d$ to be an integer and choosing the $d$ that saturates the inequality completes the proof.
\end{proof}
\noindent Theorem \ref{thm:hamsim2} then immediately follows from the Lemma by applying the art of \cite{berry2024doubling}.

\bibliography{bib}

\end{document}